\documentclass[conference]{IEEEtran}

\usepackage{amsmath,amssymb,amsthm,bm,mathtools}
\usepackage{graphicx}
\usepackage{cite}
\usepackage[bookmarks=false]{hyperref}
\usepackage{float}
\usepackage{xcolor}
\usepackage{microtype}
\usepackage{eso-pic}
\AddToShipoutPictureBG*{%
  \AtPageLowerLeft{\put(0,18){\hspace*{0.75in}\parbox{\dimexpr\paperwidth-1.5in}{\centering\scriptsize
  This work has been accepted for presentation at the 2026 IEEE Global Communications Conference (GLOBECOM).\\[2pt]
  \copyright~2026 IEEE. Personal use of this material is permitted. Permission from IEEE must be obtained for all other uses, in any current or future media, including reprinting/republishing this material for advertising or promotional purposes, creating new collective works, for resale or redistribution to servers or lists, or reuse of any copyrighted component of this work in other works.}}}}

\hypersetup{hidelinks}
\allowdisplaybreaks

\newtheorem{theorem}{Theorem}

\newtheorem{proposition}{Proposition}
\newtheorem{corollary}{Corollary}

\newtheorem{approximation}{Approximation}
\newtheorem{remark}{Remark}

\newcommand{\tr}{\mathrm{tr}}

\IEEEoverridecommandlockouts

\begin{document}
\bstctlcite{IEEEexample:BSTcontrol}

\title{Feasibility of Common-Stream-Only Sensing
in ZF-RSMA with Imperfect CSIT}

\author{
\IEEEauthorblockN{Surya Harsha and Shankar Prakriya}
\IEEEauthorblockA{Department of Electrical Engineering, Indian Institute of Technology Delhi, New Delhi, India\\
Email: eez248398@iitd.ac.in, shankar@ee.iitd.ac.in}
}

\maketitle

\begin{abstract}
Reusing the Rate-Splitting Multiple Access (RSMA) common stream for sensing eliminates dedicated probing signals in Integrated Sensing and Communication (ISAC), but creates a feasibility bottleneck: the common-power fraction must simultaneously support common decoding, private decoding, and target detection. This paper investigates this tradeoff in downlink zero-forcing RSMA with imperfect channel state information at the transmitter (CSIT) and imperfect successive interference cancellation (SIC). With the common beam fixed over a sensing dwell through sensing geometry, the common-power fraction becomes the sole design variable. We prove that the communication-feasible set is either a closed interval or empty. Analytical outage expressions localize these endpoints, revealing a high-signal-to-noise ratio (SNR) floor-ceiling effect: common decoding dictates a common-power floor, while imperfect CSIT and SIC impose a private-decoding ceiling. For sensing, the conditional detection constraint is affine, preserving the interval structure. A geometry-matched common beam ensures monotonic detection improvement, reducing joint-feasibility check to the communication-feasible upper endpoint. Simulations validate our analysis, identifying regimes where common-stream-only sensing is viable.
\end{abstract}


\section{Introduction}

Integrated Sensing and Communication (ISAC) is a key 6G capability in which spatial, spectral, and power resources are shared across data transmission and environmental sensing \cite{Clerckx2024ProcIEEE}. Rate-Splitting Multiple Access (RSMA) is a natural fit for ISAC because it manages interference via common/private message decomposition. As the common stream is decoded by all users and beamformed using slow-timescale channel information \cite{Clerckx2023JSAC}, it can serve as a dual-purpose sensing waveform, eliminating the need for dedicated probing signals \cite{Xu2021JSTSP,Yin2022ComLetters}. 

However, recent experiments expose a practical common-decoding bottleneck under finite-modulation-and-coding scheme (MCS) operation \cite{Lyu2026TWC}. In this common-stream-only sensing regime, the common-power fraction $\alpha \in [0,1]$ becomes the sole fast-timescale operating variable: increasing $\alpha$ strengthens common decoding and target detection, but invariably reduces private-stream power and, under hardware impairments, amplifies residual interference. The key question is, therefore feasibility: does any $\alpha$ simultaneously satisfy both communication and sensing requirements?

This question connects to a wider RSMA-ISAC literature that has primarily studied joint precoder design and rate-sensing tradeoffs under varying degrees of channel state information at the transmitter (CSIT) \cite{Xu2021JSTSP,CernaLoli2022arXiv}. While recent works address non-ideal operations, such as imperfect CSIT and successive interference cancellation (SIC) ~\cite{Liu2026TWC, Cheng2025Phycom}, they focus on optimizing selected designs rather than characterizing the admissible common-power set itself. Because the scalar $\alpha$ dictates conflicting system dynamics, the issue is not merely identifying the optimal operating point, but determining whether a viable one exists at all and, what structural limits affect it.

This feasibility problem becomes strictly bounded under imperfect CSIT and SIC, precisely the non-ideal conditions RSMA is designed to handle \cite{Dai2016TWC}. To study this in a reduced-order yet meaningful regime, we consider finite-dimensional downlink zero-forcing (ZF) RSMA with common-stream-only sensing under these non-ideal conditions. By fixing the common beam based on slow-timescale sensing geometry, we isolate the fast-timescale design to the scalar $\alpha$.

\noindent Our main contributions are:
\begin{enumerate}
    \item \textbf{Communication Interval Law (Theorem~\ref{thm:com_interval}):} The communication-feasible set for $\alpha$ is a closed interval or empty, proven independently of later approximations.
    \item \textbf{Finite-SNR Boundaries and High-SNR Frontier (Proposition~\ref{prop:analytical_outages}, Corollary~\ref{cor:high_snr_frontier}):} Outage expressions localize finite-SNR boundaries; their high-SNR limits reveal a nonzero common-decoding floor on $\alpha$ and a private-stream ceiling persisting as transmit power $P \to \infty$.
    \item \textbf{Conditional Joint Feasibility (Theorem~\ref{thm:joint_interval}, Corollary~\ref{cor:matched_geometry}):} Any fixed beam retains a joint-feasible interval. Furthermore, a geometry-matched common beam ensures monotonic detection, reducing the feasibility check to the upper communication-feasible endpoint. 
\end{enumerate}

Monte Carlo simulations validate the analysis and distinguish resolvable finite-SNR infeasibility from fundamental high-SNR structural limits.

\section{System Model}
\label{sec:system_model}

\begin{figure}[t]
    \centering
    \includegraphics[width=\columnwidth]{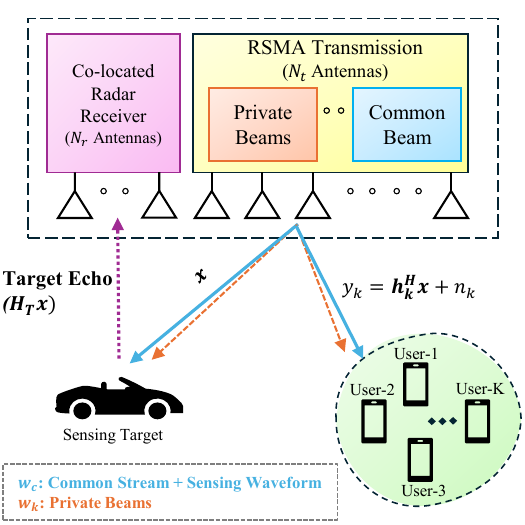}
    \caption{Downlink ZF-RSMA ISAC system model.}
    \label{fig:sys_model}
\end{figure}

We consider a downlink RSMA--ISAC system in which a base station (BS) with
\(N_t\) transmit antennas serves \(K(\leq N_t)\) single-antenna users
while illuminating one sensing target, as shown in Fig.~\ref{fig:sys_model}. The analysis is for one sensing dwell after a coarse sensing stage has provided a nominal target cell or geometry. Over this dwell, the common beam is fixed, the private ZF layer is updated from imperfect CSIT, and \(\alpha\in[0,1]\) denotes the common-power fraction.

\subsection{Communication Model}

Let $\mathbf{h}_k\in\mathbb{C}^{N_t\times 1}$ denote the channel from
the BS to user $k$. Under the Gauss--Markov imperfect-CSIT model
\cite{Dai2016TWC},
\begin{equation}
\mathbf{h}_k
=
\sqrt{1-\tau^2}\,\hat{\mathbf{h}}_k
+
\tau\mathbf{e}_k,
\label{eq:gm_model}
\end{equation}
where $\hat{\mathbf{h}}_k,\mathbf{e}_k\sim
\mathcal{CN}(\mathbf{0},\sigma_h^2\mathbf{I}_{N_t})$ are mutually
independent and independent across users, and $\tau\in[0,1]$ quantifies
the CSIT imperfection.

Let $s_c$ and $\{s_k\}_{k=1}^K$ be mutually independent unit-power
symbols, where $s_c$ is reused as the sensing waveform and no dedicated probing signal is transmitted. With total transmit power $P$, the private-stream power is $P_p=(1-\alpha)P/K$. The common beam
$\mathbf{w}_c\in\mathbb{C}^{N_t\times 1}$ is unit norm and fixed over
the sensing dwell. For the private layer, define
$\hat{\mathbf{H}}\triangleq
[\hat{\mathbf{h}}_1,\ldots,\hat{\mathbf{h}}_K]
\in\mathbb{C}^{N_t\times K}$. The BS observes $\hat{\mathbf{H}}$, not the true channel matrix, and constructs the
unit-norm ZF private beam directions, so that $P_p$ remains the per-stream power,
from $\hat{\mathbf{H}}$ as in~\cite{Wiesel2008TSP}
\begin{equation}
\widetilde{\mathbf{W}}
\triangleq
[\widetilde{\mathbf{w}}_1,\ldots,\widetilde{\mathbf{w}}_K]
=
\hat{\mathbf{H}}(\hat{\mathbf{H}}^H\hat{\mathbf{H}})^{-1},
\
\mathbf{w}_k=
\frac{\widetilde{\mathbf{w}}_k}
{\|\widetilde{\mathbf{w}}_k\|},
\label{eq:zf_precoder}
\end{equation}
for $k=1,\ldots,K$, assuming $\hat{\mathbf{H}}$ has full column rank.
Let $\mathbf{W} \triangleq [\mathbf{w}_1,\ldots,\mathbf{w}_K]$; then $\hat{\mathbf{h}}_i^H\mathbf{w}_k=0$ for $i\neq k$, but the true-channel cross terms are $\mathbf{h}_i^H\mathbf{w}_k=\tau\mathbf{e}_i^H\mathbf{w}_k$ for $i\neq k$; hence imperfect CSIT converts nominal ZF nulls into residual multiuser leakage. The transmitted signal is

\vspace{-3mm}
\begin{equation}
\mathbf{x}
=
\sqrt{\alpha P}\,\mathbf{w}_c s_c
+
\sum_{k=1}^{K}\sqrt{P_p}\,\mathbf{w}_k s_k,
\label{eq:tx_signal}
\end{equation}

\noindent and the received signal at user $k$ is
$y_k
=
\mathbf{h}_k^H\mathbf{x}+n_k,
n_k\sim\mathcal{CN}(0,\sigma_n^2).$
Each user first decodes the common stream, performs SIC, and then
decodes its private stream \cite{Clerckx2023JSAC}. The signal-to-interference-plus-noise ratios (SINRs) below assume the standard RSMA receiver-side knowledge for decoding and SIC.
Imperfect SIC is modeled by a residual factor $\delta\in[0,1]$, where
$\delta=0$ corresponds to ideal SIC. The resulting common-stream and
private-stream SINRs at user $k$ are
\begin{equation}
\gamma_{c,k}
=
\frac{\alpha P |\mathbf{h}_k^H\mathbf{w}_c|^2}
{\sum_{j=1}^{K} P_p |\mathbf{h}_k^H\mathbf{w}_j|^2+\sigma_n^2},
\label{eq:sinr_common}
\end{equation}
and
\vspace{-2mm}
\begin{equation}
\gamma_{p,k}
=
\frac{P_p|\mathbf{h}_k^H\mathbf{w}_k|^2}
{\sum_{\substack{j=1\\ j\neq k}}^{K}
P_p|\mathbf{h}_k^H\mathbf{w}_j|^2
+
\delta\alpha P|\mathbf{h}_k^H\mathbf{w}_c|^2
+
\sigma_n^2}.
\label{eq:sinr_private}
\end{equation}

\subsection{Sensing-Detection Model}

For sensing, we adopt a single-target co-located Multiple-Input Multiple-Output (MIMO) radar model over
$L$ snapshots. The radar receiver has $N_r$ antennas, and the target
response matrix is $\mathbf{H}_{\mathrm{T}}\in\mathbb{C}^{N_r\times N_t}$. Clutter and receiver noise are modeled through a positive-definite covariance $\mathbf{R}\in\mathbb{C}^{N_r\times N_r}$, assumed known or pre-estimated over the dwell; the clutter-free/noise-only model is modelled by $\mathbf{R}=\sigma_r^2\mathbf{I}_{N_r}$.
Define
$\mathbf{Q}
\triangleq
\mathbf{H}_{\mathrm{T}}^H
\mathbf{R}^{-1}
\mathbf{H}_{\mathrm{T}},
\mathbf{Q}\succeq\mathbf{0}.$
The nominal dwell-level sensing geometry is therefore encoded by $\mathbf{H}_{\mathrm{T}}$ and $\mathbf{R}$.
Since the common stream is the only RSMA layer intentionally shaped for
sensing, the slow-timescale common beam is chosen as
\vspace{-1.5mm}
\begin{equation}
\mathbf{w}_c
\in
\operatorname*{arg\,max}_{\|\mathbf{u}\|=1}
\mathbf{u}^H\mathbf{Q}\mathbf{u},
\label{eq:wc_geometry_design}
\vspace{-0.5mm}
\end{equation}
i.e., as a dominant eigenvector of $\mathbf{Q}$. This beam is fixed over
the fast user-fading ensemble, whereas the private beams in
\eqref{eq:zf_precoder} vary with the imperfect CSIT.

Condition on a realized private beam set
$\mathbf{W}$. Define the aggregate spatial covariance
$\mathbf{B}_{\mathbf{W}}
\triangleq
\frac{1}{K}
\sum_{k=1}^{K}
\mathbf{w}_k\mathbf{w}_k^H .$
the transmit covariance is
\begin{equation}
\mathbf{R}_x(\alpha;\mathbf{W})
=
P\left(
\alpha\mathbf{w}_c\mathbf{w}_c^H
+
(1-\alpha)\mathbf{B}_{\mathbf{W}}
\right).
\label{eq:tx_cov}
\end{equation}
For a radar processor that knows $\mathbf{R}$, the transmitted common waveform, and $\mathbf{R}_x(\alpha;\mathbf{W})$, the large-snapshot General Likelihood Ratio Test (GLRT) model for single-target co-located MIMO radar
with known clutter covariance \cite{LiStoica2009Book} gives the
conditional noncentrality parameter
\begin{equation}
\Lambda_{\mathbf{W}}(\alpha)
=
2L\,\operatorname{tr}
\left(
\mathbf{Q}\mathbf{R}_x(\alpha;\mathbf{W})
\right).
\label{eq:lambda_trace}
\end{equation}
Define
$q_c
\triangleq
\mathbf{w}_c^H\mathbf{Q}\mathbf{w}_c,
q_p(\mathbf{W})
\triangleq
\operatorname{tr}
\left(
\mathbf{Q}\mathbf{B}_{\mathbf{W}}
\right),
\Delta_{\mathbf{W}}
\triangleq
q_c-q_p(\mathbf{W}).$
Using the linearity of the trace and the cyclic identity
$\operatorname{tr}(\mathbf{Q}\mathbf{u}\mathbf{u}^{H})
= \operatorname{tr}(\mathbf{u}^{H}\mathbf{Q}\mathbf{u})
= \mathbf{u}^{H}\mathbf{Q}\mathbf{u} (\forall \mathbf{u})$, substituting \eqref{eq:tx_cov} into \eqref{eq:lambda_trace} gives
\vspace{-3mm}
\begin{equation}
\Lambda_{\mathbf{W}}(\alpha)
=
2LP
\left[
q_p(\mathbf{W})
+
\alpha\Delta_{\mathbf{W}}
\right].
\label{eq:lambda_affine}
\end{equation}
For a detection threshold $\eta>0$ chosen to satisfy the desired
false-alarm constraint, the corresponding conditional detection
probability is
\vspace{-3mm}
\begin{equation}
\bar{P}_{\mathrm{D},\mathbf{W}}(\alpha)
\triangleq
Q_d
\left(
\sqrt{\Lambda_{\mathbf{W}}(\alpha)},
\sqrt{\eta}
\right),
\label{eq:pd_conditional}
\end{equation}
where $Q_d(\cdot,\cdot)$ is the generalized Marcum $Q$-function of order
$d$, with $d$ determined by the degrees of freedom of the GLRT statistic
\cite{Shnidman1989TIT,LiStoica2009Book}.

\section{Communication-Feasible Interval}
\label{sec:com_feasible}

With the slow-timescale common beam $\mathbf{w}_c$ and transmit power $P$
fixed, the common-power fraction $\alpha\in[0,1]$ is the only remaining
operating variable. Let the SINR thresholds be $T_c=2^{R_c}-1$ and $T_p=2^{R_p}-1$, and define the
representative per-user layer outage probabilities
$
p_{c,u}(\alpha)=\Pr\{\gamma_{c,k}(\alpha)<T_c\},\quad
p_{p,u}(\alpha)=\Pr\{\gamma_{p,k}(\alpha)<T_p\}.
$
Under the statistically symmetric user model, the multicast common-stream
constraint is enforced through $\epsilon_c=\min\{\epsilon_c^{\mathrm{u}},\epsilon_c^{\mathrm{sys}}/K\}$ (the per-user and system-level outage targets), since the union bound gives
$p_{c,\mathrm{sys}}(\alpha)\le Kp_{c,u}(\alpha)$. Hence, for a given $P$, the
communication-feasible common-power set is
\begin{equation}
    \mathcal{A}_{\mathrm{comm}}(P)=
\left\{\alpha\in[0,1]:
p_{c,u}(\alpha)\le\epsilon_c,\;
p_{p,u}(\alpha)\le\epsilon_p^{\mathrm{u}}
\right\}.
\label{eq:Acomm_def}
\end{equation}

This layer-wise formulation follows the RSMA decoding order.
It also gives a conservative end-to-end guarantee. For
\[
\mathcal{O}_k^{\mathrm{e2e}}(\alpha)
=
\{\gamma_{c,k}(\alpha)<T_c\}
\cup
\{\gamma_{c,k}(\alpha)\ge T_c,\gamma_{p,k}(\alpha)<T_p\},
\]
we have $p_k^{\mathrm{e2e}}(\alpha)
\le p_{c,u}(\alpha)+p_{p,u}(\alpha).$
Thus any $\alpha\in\mathcal{A}_{\mathrm{comm}}(P)$ ensures
$p_k^{\mathrm{e2e}}(\alpha)\le\epsilon_k^{\mathrm{e2e}}$ whenever
$\epsilon_c+\epsilon_p^{\mathrm{u}}\le\epsilon_k^{\mathrm{e2e}}$.

We keep the two layer constraints rather than replacing them with a single
end-to-end outage constraint because the feasible geometry in $\alpha$ is
two-sided. For each fading realization, common decoding imposes a lower
threshold on $\alpha$, while private decoding imposes an upper threshold:
$
\{\gamma_{c,k}(\alpha)\ge T_c,\gamma_{p,k}(\alpha)\ge T_p\}
=
\{\alpha_c(\omega)\le\alpha\le\alpha_p(\omega)\},
$
with an empty success interval if $\alpha_c(\omega)>\alpha_p(\omega)$.
Averaging over fading therefore yields the coverage probability of random
intervals on the $\alpha$-axis, which need not be monotone or quasi-concave.
The set $\mathcal{A}_{\mathrm{comm}}(P)$ preserves this RSMA structure,
captures the opposing effects of $\alpha$, and motivates the interval law
derived next.

\begin{theorem}[Communication interval law]
\label{thm:com_interval}
For every fixed $P$ and fixed unit-norm common beam $\mathbf{w}_c$, the
communication-feasible set $\mathcal{A}_{\mathrm{comm}}(P)$ is either a
closed interval $[\alpha_{\min}(P),\alpha_{\max}(P)]\subseteq[0,1]$ or the
empty set.
\end{theorem}

\begin{IEEEproof}
Fix \(P\) and \(\mathbf w_c\). For each
user \(k\), define
$a_k=P|\mathbf h_k^{\rm H}\mathbf w_c|^2,
b_k=\frac{P}{K}\sum_{j=1}^{K}|\mathbf h_k^{\rm H}\mathbf w_j|^2,
d_k=\frac{P}{K}|\mathbf h_k^{\rm H}\mathbf w_k|^2,
e_k=b_k-d_k .$
Then
\[
\gamma_{c,k}(\alpha)=
\frac{\alpha a_k}{(1-\alpha)b_k+\sigma_n^2},
\gamma_{p,k}(\alpha)=
\frac{(1-\alpha)d_k}{(1-\alpha)e_k+\alpha\delta a_k+\sigma_n^2},
\]
and direct differentiation gives $
\frac{d\gamma_{c,k}}{d\alpha}
\ge0,\
\frac{d\gamma_{p,k}}{d\alpha}\le0.$
Thus, for every channel realization, the common SINRs are nondecreasing in
\(\alpha\), while the private SINRs are nonincreasing. Hence the corresponding
common and private outage probabilities \(p_{c,u}(\alpha)\) and
\(p_{p,u}(\alpha)\) are respectively nonincreasing and nondecreasing.

Moreover, \(\gamma_{c,k}(\alpha)\) and \(\gamma_{p,k}(\alpha)\) are continuous
in \(\alpha\). For any \(\alpha_n\to\alpha\), the relevant outage indicators
therefore converge almost surely, except possibly on the threshold-equality
events \(\{\gamma_{c,k}(\alpha)=T_c\}\) and
\(\{\gamma_{p,k}(\alpha)=T_p\}\), which have probability zero by the absolute
continuity of the channel law. Dominated convergence then yields continuity of
\(p_{c,u}\) and \(p_{p,u}\) on \([0,1]\).

Consequently, the common constraint is a closed upper subset of \([0,1]\), and
the private constraint is a closed lower subset of \([0,1]\). Their intersection
is therefore either empty or a closed interval
\([\alpha_{\min}(P),\alpha_{\max}(P)]\subseteq[0,1]\).
\end{IEEEproof}

\begin{remark}[Structural generality]
\label{rem:com_generality}
Theorem~\ref{thm:com_interval} is a sample-path result and is independent of the particular fixed unit-norm common beamformer $\mathbf{w}_c$, including \eqref{eq:wc_geometry_design}. It is also not tied to equal private powers: for any fixed split $P_j=(1-\alpha)P\beta_j$, $\beta_j\ge0$, $\sum_{j=1}^{K}\beta_j=1$, with the equal split recovered by $\beta_j=1/K$, the same monotonicity gives the interval law. User-specific outage targets are also allowed; symmetry is used only to express the bounds through $p_{c,u}(\alpha)$ and $p_{p,u}(\alpha)$.
\end{remark}

\section{Analytical Boundary Localization and High-SNR Frontier}
\label{sec:outage}

To localize $\alpha_{\min}(P)$ and $\alpha_{\max}(P)$, we adopt tractable approximations for $p_{c,u}(\alpha)$ and $p_{p,u}(\alpha)$. As noted in Section~\ref{sec:system_model}, the private-layer interference terms are the CSIT-error-induced ZF leakages. Because $\mathbf{h}_k$ is isotropically distributed and $\|\mathbf{w}_c\|=1$, the distribution of $X_k=|\mathbf{h}_k^H\mathbf{w}_c|^2$ is invariant to the slow-timescale choice of $\mathbf{w}_c$; consequently, the communication-side analytical localization derived below remains valid after common-beam design. Let $M \triangleq N_t-K+1$ and $m \triangleq K-1$ (reflecting the $K-1$ interfering private beams). Due to imperfect CSIT ZF alignment, the common gain $X_k \triangleq |\mathbf{h}_k^H\mathbf{w}_c|^2 \sim \mathrm{Exp}(\sigma_h^2)$, the intended private gain $Y_k \triangleq |\mathbf{h}_k^H\mathbf{w}_k|^2$ follows the Gamma mixture $\sum_{n=1}^{M}\pi_n\Gamma(n,\sigma_h^2)$ with weights $\pi_n \triangleq \binom{M-1}{n-1}(1-\tau^2)^{n-1}(\tau^2)^{M-n}$ ~\cite{Dai2016TWC,Moschopoulos1985AISM}, and, since $\mathbf{h}_k^H\mathbf{w}_j=\tau\mathbf{e}_k^H\mathbf{w}_j$ for $j\neq k$, the individual ZF leakages satisfy $|\mathbf{h}_k^H\mathbf{w}_j|^2 \sim \mathrm{Exp}(\tau^2\sigma_h^2).$ 

The aggregate leakage
$Z_k=\tau^2\mathbf{e}_k^H\mathbf{A}_k\mathbf{e}_k$, where
$\mathbf{A}_k\triangleq \sum_{j\neq k}\mathbf{w}_j\mathbf{w}_j^H \succeq \mathbf{0}$,
is a positive Gaussian quadratic form. If $\{\lambda_{k,r}\}$ are the nonzero eigenvalues of $\mathbf{A}_k$, then conditioned on $\mathbf{W}$,
$Z_k \overset{d}{=} \tau^2\sigma_h^2\sum_r \lambda_{k,r}E_r$, $E_r\sim \mathrm{Exp}(1)$, so $\mathcal{L}_{Z_k|\mathbf{W}}(s)=\prod_r(1+s\tau^2\sigma_h^2\lambda_{k,r})^{-1}$~\cite{AlNaffouri2016TCOM,Moschopoulos1985AISM}. To retain compact $\alpha$-dependent outage expressions, we use the following first-moment-preserving Gamma approximation.

\begin{approximation}[Aggregate leakage approximation]
\label{approx:leakage}
$Z_k^{\mathrm{ana}}\sim \Gamma(m,\tau^2\sigma_h^2)$, i.e., a trace-preserving equal-eigenvalue approximation satisfying $\mathbb{E}[Z_k\mid \mathbf{W}]=\tau^2\sigma_h^2\operatorname{tr}(\mathbf{A}_k)=m\tau^2\sigma_h^2$. This approximation is exact when the interference beams $\{\mathbf{w}_j\}_{j\neq k}$ are orthonormal, so that $\mathbf{A}_k$ is a rank-$m$ projector, and is used only to localize the outage boundaries, not Theorem~\ref{thm:com_interval}.
\end{approximation}

\noindent To evaluate outage under this approximation, we first parameterize the normalized SINR thresholds and interference powers. Let 
$s_c(\alpha) \triangleq \frac{T_c}{\alpha P \sigma_h^2}$ 
and 
$s_p(\alpha) \triangleq \frac{T_p K}{(1-\alpha)P\sigma_h^2}$ 
be the scaled common and private thresholds. Then, define the effective leakage power as $\theta_L(\alpha)\triangleq {(1-\alpha)P\tau^2\sigma_h^2}/{K}$ and the residual SIC power as $\theta_{\mathrm{res}}(\alpha)\triangleq \delta\alpha P\sigma_h^2$.

\begin{proposition}[Finite-SNR analytical success probabilities]
\label{prop:analytical_outages}
Under the above approximation, the per-user common-stream success probability is
\begin{align}
P_{c,\mathrm{suc},u}^{\mathrm{ana}}(\alpha)
&= \exp\!\big(-s_c(\alpha)\sigma_n^2\big)\left(1+\frac{T_c(1-\alpha)\tau^2}{K\alpha}\right)^{M-K}\nonumber\\
&\quad\times \left(1+\frac{T_c(1-\alpha)}{K\alpha}\right)^{-M},
\label{eq:com_succ}
\end{align}
and the per-user private-stream success probability is
\begin{align}
P_{p,\mathrm{suc},u}^{\mathrm{ana}}(\alpha)
&= \sum_{n=1}^{M}\pi_n e^{-s_p(\alpha)\sigma_n^2}\sum_{i=0}^{n-1}\frac{s_p(\alpha)^i}{i!}\nonumber\\
&\quad\times \sum_{\ell=0}^{i}\binom{i}{\ell}(\sigma_n^2)^{i-\ell}\,\mathcal{T}_{\ell}\!\big(s_p(\alpha);\alpha\big),
\label{eq:pri_succ}
\end{align}
where
$
\mathcal{T}_{\ell}(s;\alpha)\triangleq
\sum_{t=0}^{\ell}\frac{\ell!}{t!}(m)_t
\frac{\theta_L(\alpha)^t\theta_{\mathrm{res}}(\alpha)^{\ell-t}}
{\bigl(1+\theta_L(\alpha)s\bigr)^{m+t}\bigl(1+\theta_{\mathrm{res}}(\alpha)s\bigr)^{1+\ell-t}},
$
and $(m)_t \triangleq \Gamma(m+t)/\Gamma(m)$.
\end{proposition}

\begin{IEEEproof}[Proof sketch]
For common decoding, conditioning on the private-layer denominator in~\eqref{eq:sinr_common} and applying the exponential tail of $X_k$, the Gamma-mixture transform of $Y_k$, and Approximation~\ref{approx:leakage} yields~\eqref{eq:com_succ}. For private decoding, the interference-plus-noise denominator can be modeled as $N_\alpha \triangleq \theta_L(\alpha)\,U+\theta_{\mathrm{res}}(\alpha)\,V+\sigma_n^2$, where independent random variables $U\sim\Gamma(m,1)$ and $V\sim\mathrm{Exp}(1)$ isolate the normalized leakage and residual SIC. Conditioning on $N_\alpha$ and expanding the Gamma-mixture tail of $Y_k$ produces terms of the form $\mathbb{E}[N_\alpha^i e^{-sN_\alpha}]=(-1)^i\mathcal{L}_{N_\alpha}^{(i)}(s)$. Differentiating the Laplace transform to collect the leakage and residual-SIC powers yields $\mathcal{T}_{\ell}(s;\alpha)$, giving~\eqref{eq:pri_succ}.
\end{IEEEproof}

\noindent The expressions above are useful for locating the endpoints at a specified transmit power, but not reveal if infeasibility is merely noise-limited or persists as $P$ grows. The high-SNR limit is therefore used as a diagnostic frontier: thermal-noise terms vanish, whereas CSIT-error leakage and residual-SIC interference remain signal-scaled. This exposes the limiting common-power floor and private-power ceiling before the sensing constraint is imposed.

\begin{corollary}[High-SNR analytical frontier]
\label{cor:high_snr_frontier}
For every fixed $\alpha\in(0,1)$, as $P\to\infty$,
$\lim_{P\to\infty}P_{c,\mathrm{suc},u}^{\mathrm{ana}}(\alpha)
= \Psi_c^{\mathrm{ana}}(\alpha),
\lim_{P\to\infty}P_{p,\mathrm{suc},u}^{\mathrm{ana}}(\alpha)
= \Psi_p^{\mathrm{ana}}(\alpha), \label{eq:Psi_p_limit}$
where
\begin{equation}
\Psi_c^{\mathrm{ana}}(\alpha)
\triangleq
\left(1+\frac{T_c(1-\alpha)\tau^2}{K\alpha}\right)^{M-K}
\left(1+\frac{T_c(1-\alpha)}{K\alpha}\right)^{-M},
\label{eq:Psi_c}
\end{equation}
and
\begin{align}
\Psi_p^{\mathrm{ana}}(\alpha)
&\triangleq \sum_{n=1}^{M}\pi_n \sum_{\ell=0}^{n-1}\sum_{t=0}^{\ell}\frac{(m)_t}{t!}\nonumber\\
&\quad\times
\frac{(T_p\tau^2)^t\left(\dfrac{T_p\delta K\alpha}{1-\alpha}\right)^{\ell-t}}
{(1+T_p\tau^2)^{m+t}\left(1+\dfrac{T_p\delta K\alpha}{1-\alpha}\right)^{1+\ell-t}}.
\label{eq:Psi_p}
\end{align}
\end{corollary}

\begin{IEEEproof}[Proof sketch]
    For fixed $\alpha\in(0,1)$, $s_c(\alpha)\sigma_n^2\to0$ and $s_p(\alpha)\sigma_n^2\to0$ as $P\to\infty$. Moreover, $s_p(\alpha)\theta_L(\alpha)\to T_p\tau^2$ and $s_p(\alpha)\theta_{\mathrm{res}}(\alpha)\to T_p\delta K\alpha/(1-\alpha)$. Substituting these limits into Proposition~\ref{prop:analytical_outages} gives~\eqref{eq:Psi_c} and~\eqref{eq:Psi_p}.
\end{IEEEproof}

Let $\eta_c\triangleq 1-\epsilon_c$, $\eta_p\triangleq 1-\epsilon_p^u$, and $
p_{c,u}^{\mathrm{ana}}(\alpha)\triangleq 1-P_{c,\mathrm{suc},u}^{\mathrm{ana}}(\alpha),\
p_{p,u}^{\mathrm{ana}}(\alpha)\triangleq 1-P_{p,\mathrm{suc},u}^{\mathrm{ana}}(\alpha).$ Define
$ A_{\mathrm{com}}^{\mathrm{ana}}(P)
\triangleq
\left\{\alpha\in[0,1]:
p_{c,u}^{\mathrm{ana}}(\alpha)\le \epsilon_c,\;
p_{p,u}^{\mathrm{ana}}(\alpha)\le \epsilon_p^u
\right\}.
\label{eq:Acom_ana_set}$
When nonempty, write
$A_{\mathrm{com}}^{\mathrm{ana}}(P)=
\bigl[\alpha_{\min}^{\mathrm{ana}}(P),\alpha_{\max}^{\mathrm{ana}}(P)\bigr].$

\begin{remark}[High-SNR interval structure]
\label{rem:asymptotic_interval}
$\Psi_c^{\mathrm{ana}}(\alpha)$ is strictly increasing on $(0,1)$ because it is strictly decreasing in $(1-\alpha)/\alpha$. Also, $\Psi_p^{\mathrm{ana}}(\alpha)$ is nonincreasing on $(0,1)$, since each term in \eqref{eq:Psi_p} is nonincreasing in $\alpha/(1-\alpha)$, and is strictly decreasing when $\delta>0$. Consequently,
\begin{equation}
A_{\infty}^{\mathrm{ana}}
\triangleq
\left\{\alpha\in(0,1):
\Psi_c^{\mathrm{ana}}(\alpha)\ge \eta_c,\;
\Psi_p^{\mathrm{ana}}(\alpha)\ge \eta_p
\right\}
\label{eq:Ainf_ana_def}
\end{equation}
is either empty or an interval; when nonempty, write
$A_{\infty}^{\mathrm{ana}}=
\bigl[\alpha_{\min,\infty}^{\mathrm{ana}},\alpha_{\max,\infty}^{\mathrm{ana}}\bigr].$
Thus, the analytical high-SNR frontier retains the floor-ceiling geometry of Theorem~\ref{thm:com_interval}: common decoding induces a floor on $\alpha$, while imperfect CSIT and SIC can maintain a private-stream ceiling.
\end{remark}

\section{Conditional Joint Feasibility}
\label{sec:joint_feasible}

Section~\ref{sec:com_feasible} showed that the communication constraints reduce the admissible common-power fractions to a closed interval $\mathcal{A}_{\mathrm{comm}}(P)\subseteq[0,1]$ or the empty set. We now characterize how target detection for a fixed realized private beam set $\mathbf{W}$ further restricts $\alpha$.

Recalling ~\eqref{eq:lambda_affine}, for a target detection probability $P_{\mathrm{D}}^{\star}\in(0,1)$, let
$\widetilde{\lambda}\triangleq\lambda^{\star}/(2LP)$, where
$\lambda^{\star}\triangleq\inf\{\lambda\ge0:Q_d(\sqrt{\lambda},\sqrt{\eta})\ge P_{\mathrm{D}}^{\star}\}$\cite{Shnidman1989TIT,LiStoica2009Book}.
Consequently, the constraint $\bar{P}_{\mathrm{D},\mathbf{W}}(\alpha)\ge P_{\mathrm{D}}^{\star}$
is equivalent to the affine inequality
$q_p(\mathbf{W})+\alpha\Delta_{\mathbf{W}}\ge\widetilde{\lambda}$.

\begin{theorem}[Conditional detection interval and endpoint rule]
\label{thm:joint_interval}
For fixed $\mathbf{W}$, the conditional detection-feasible set
$\mathcal{A}_{\mathrm{det},\mathbf{W}}(P_{\mathrm{D}}^{\star})
\triangleq\{\alpha\in[0,1]:\bar{P}_{\mathrm{D},\mathbf{W}}(\alpha)\ge P_{\mathrm{D}}^{\star}\}$
is either an interval or the empty set. Explicitly,
\begin{equation}
\mathcal{A}_{\mathrm{det},\mathbf{W}}(P_{\mathrm{D}}^{\star})
=
\begin{cases}
[\alpha_{\mathrm{det},\mathbf{W}},1]\cap[0,1], & \Delta_{\mathbf{W}}>0,\\[2pt]
[0,1], & \Delta_{\mathbf{W}}=0,\;q_p(\mathbf{W})\ge\widetilde{\lambda},\\[2pt]
\varnothing, & \Delta_{\mathbf{W}}=0,\;q_p(\mathbf{W})<\widetilde{\lambda},\\[2pt]
[0,\alpha_{\mathrm{det},\mathbf{W}}]\cap[0,1], & \Delta_{\mathbf{W}}<0,
\end{cases}
\label{eq:Adet_cases}
\end{equation}
where $\alpha_{\mathrm{det},\mathbf{W}}
\triangleq
\frac{\widetilde{\lambda}-q_p(\mathbf{W})}{\Delta_{\mathbf{W}}}.$
Consequently,
$\mathcal{A}_{\mathrm{joint},\mathbf{W}}(P_{\mathrm{D}}^{\star})
\triangleq
\mathcal{A}_{\mathrm{comm}}(P)\cap
\mathcal{A}_{\mathrm{det},\mathbf{W}}(P_{\mathrm{D}}^{\star})$
is an interval or the empty set. Moreover, if
$\mathcal{A}_{\mathrm{comm}}(P)=[\alpha_{\min}(P),\alpha_{\max}(P)]\neq\varnothing$, then
\begin{equation}
\operatorname*{arg\,max}_{\alpha\in\mathcal{A}_{\mathrm{comm}}(P)}
\bar{P}_{\mathrm{D},\mathbf{W}}(\alpha)
=
\begin{cases}
\{\alpha_{\max}(P)\}, & \Delta_{\mathbf{W}}>0,\\[2pt]
\mathcal{A}_{\mathrm{comm}}(P), & \Delta_{\mathbf{W}}=0,\\[2pt]
\{\alpha_{\min}(P)\}, & \Delta_{\mathbf{W}}<0.
\end{cases}
\label{eq:endpoint_rule}
\end{equation}
\end{theorem}

\begin{IEEEproof}
By~\eqref{eq:lambda_affine}, $\Lambda_{\mathbf{W}}(\alpha)$ is affine in $\alpha$ with slope $2LP\Delta_{\mathbf{W}}$. Since the function $Q_d(\sqrt{\lambda},\sqrt{\eta})$ is continuous and nondecreasing for $\lambda\ge0$, the detection constraint reduces to an affine inequality in $\alpha$, which yields~\eqref{eq:Adet_cases}. Theorem~\ref{thm:com_interval} gives $\mathcal{A}_{\mathrm{comm}}(P)$ as an interval or the empty set, so their intersection is likewise an interval or the empty set. The endpoint rule in~\eqref{eq:endpoint_rule} follows from the monotonicity of $\bar{P}_{\mathrm{D},\mathbf{W}}(\alpha)$ over $\mathcal{A}_{\mathrm{comm}}(P)$.
\end{IEEEproof}

\begin{corollary}[Matched-geometry common beam]
\label{cor:matched_geometry}
If $\mathbf{w}_c$ satisfies~\eqref{eq:wc_geometry_design}, then $\Delta_{\mathbf{W}}\ge 0$ for every realized $\mathbf{W}$. Consequently, $\bar{P}_{\mathrm{D},\mathbf{W}}(\alpha)$ is nondecreasing in $\alpha$, the right endpoint $\alpha_{\max}(P)$ is always sensing-optimal over any nonempty $\mathcal{A}_{\mathrm{comm}}(P)$, and the conditional joint-feasibility test reduces to $
\bar{P}_{\mathrm{D},\mathbf{W}}\!\bigl(\alpha_{\max}(P)\bigr)\ge P_{\mathrm{D}}^{\star}.$
In particular, if this fails, then no $\alpha\in\mathcal{A}_{\mathrm{comm}}(P)$ is jointly feasible.
\end{corollary}

\begin{IEEEproof}
Under~\eqref{eq:wc_geometry_design}, $q_c=\lambda_{\max}(\mathbf{Q})$. Also, $\mathbf{B}_{\mathbf{W}}\succeq\mathbf{0}$ and
$\tr(\mathbf{B}_{\mathbf{W}})
=
\frac{1}{K}\sum_{k=1}^{K}\tr(\mathbf{w}_k\mathbf{w}_k^{H})
=
1,$ so
$
q_p(\mathbf{W})
=
\tr(\mathbf{Q}\mathbf{B}_{\mathbf{W}})
\le
\lambda_{\max}(\mathbf{Q})\tr(\mathbf{B}_{\mathbf{W}})
=
q_c.$
Hence $\Delta_{\mathbf{W}}\ge 0$. The remaining claims follow from Theorem~\ref{thm:joint_interval}.
\end{IEEEproof}

\section{Numerical Results}

This section validates the analytical communication model, illustrates the high-SNR feasibility frontier, and verifies the sensing-side endpoint rule. Unless otherwise stated, $(N_t,K,\tau,\delta,\mathrm{SNR},R_c,R_p)=(8,4,0.20,0.05,15\,\mathrm{dB},0.20,1.00)$, $\sigma_h^2=\sigma_n^2=1$, and $(\epsilon_c^u,\epsilon_c^{\mathrm{sys}},\epsilon_p^u)=(0.10,0.40,0.10)$, so that $\epsilon_c=\min\{\epsilon_c^u,\epsilon_c^{\mathrm{sys}}/K\}=0.10$. All the Monte Carlo results use $2\times 10^5$ channel realizations. For sensing, $L=16$, and the finite-$L$ GLRT estimates are obtained by Monte Carlo over the transmitted snapshots. Since the user channels are isotropic, the communication-side results in Figs.~\ref{fig:baseline_outage}--\ref{fig:highsnr_frontier} depend only on the fact that $\mathbf{w}_c$ is deterministic and unit norm, not on its particular direction. Fig.~\ref{fig:detection_audit} is therefore the only experiment that explicitly tests the common-beam design.

\begin{figure}[b]
    \centering
    \includegraphics[width=0.9\columnwidth]{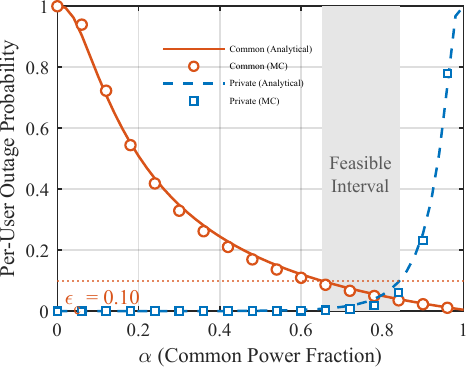}
    \caption{Layer-specific per-user common-stream and private-stream outage probabilities versus the common-power fraction $\alpha$ at the baseline operating point. Curves show the analytical outage expressions, markers show Monte Carlo estimates, and the shaded region indicates $\mathcal{A}_{\mathrm{comm}}^{\mathrm{ana}}(P)$.}
    \label{fig:baseline_outage}
\end{figure}

\begin{figure}[t]
    \centering
    \includegraphics[width=0.9\columnwidth]{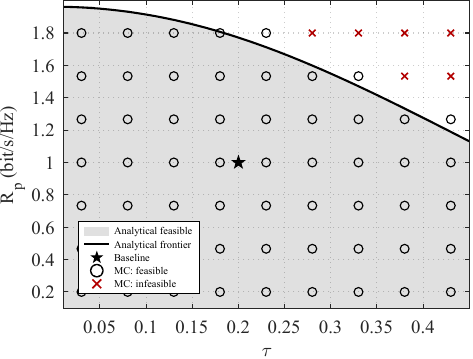}
    \caption{Communication-feasibility screening in the $(\tau,R_p)$ plane. Shading and boundary are determined by $\mathcal{A}_{\mathrm{comm}}^{\mathrm{ana}}(P)$; markers show Monte Carlo certification on the sampled $7\times 9$ grid.}
    \label{fig:frontier_2d}
\end{figure}

Figs.~\ref{fig:baseline_outage} and \ref{fig:frontier_2d} validate the finite-SNR model. In Fig.~\ref{fig:baseline_outage}, the analytical curves track Monte Carlo closely and recovers the interval structure of Theorem~\ref{thm:com_interval}: $\mathcal{A}_{\mathrm{comm}}^{\mathrm{ana}}(P)=[0.6504,\,0.8423]$, versus the Monte Carlo estimate $[0.6241,\,0.8617]$. The mismatch is conservative, producing false-infeasible rather than false-feasible decisions. Fig.~\ref{fig:frontier_2d} shows the same behavior in parameter-space screening: over the sampled grid, the analytical screen returns $54$ true-feasible points, $6$ true-infeasible points, and $3$ conservative false-infeasible points, with no false-feasible cases. Hence, $\mathcal{A}_{\mathrm{comm}}^{\mathrm{ana}}(P)$ is a conservative inner approximation to the Monte Carlo communication-feasible region on the sampled grid.

\begin{figure}[t]
    \centering
    \includegraphics[width=0.9\columnwidth]{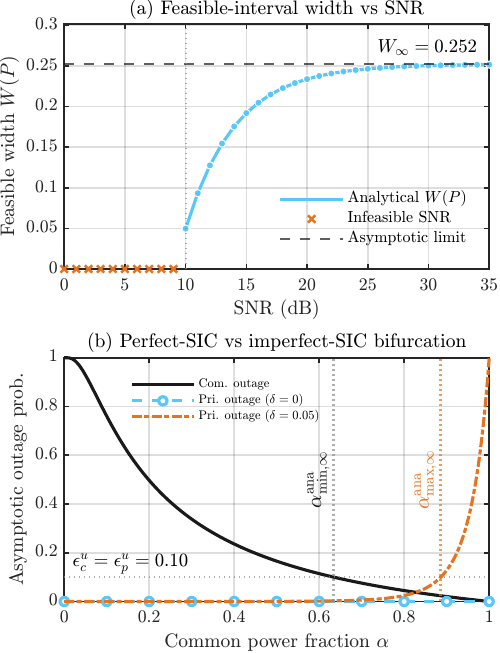}
    \caption{High-SNR analytical communication frontier under the RSMA sequential-decoding reliability constraints at the baseline operating point. (a) Width $[\alpha_{\max}^{\mathrm{ana}}(P)-\alpha_{\min}^{\mathrm{ana}}(P)]_+$ versus SNR. (b) High-SNR common-stream and private-stream outage probabilities under ideal and imperfect SIC.}
    \label{fig:highsnr_frontier}
\end{figure}

Fig.~\ref{fig:highsnr_frontier} shows that the communication obstruction is not merely a low-SNR effect. The feasible interval is empty at low SNR, opens once the private and common constraints become simultaneously satisfiable, and then saturates to a nonzero high-SNR width, validating the floor-ceiling mechanism predicted by Corollary~\ref{cor:high_snr_frontier}. With ideal SIC, this upper bound disappears, confirming that imperfect SIC is the source of the high-SNR private-stream ceiling.

\begin{figure}[t]
    \centering
    \includegraphics[width=0.9\columnwidth]{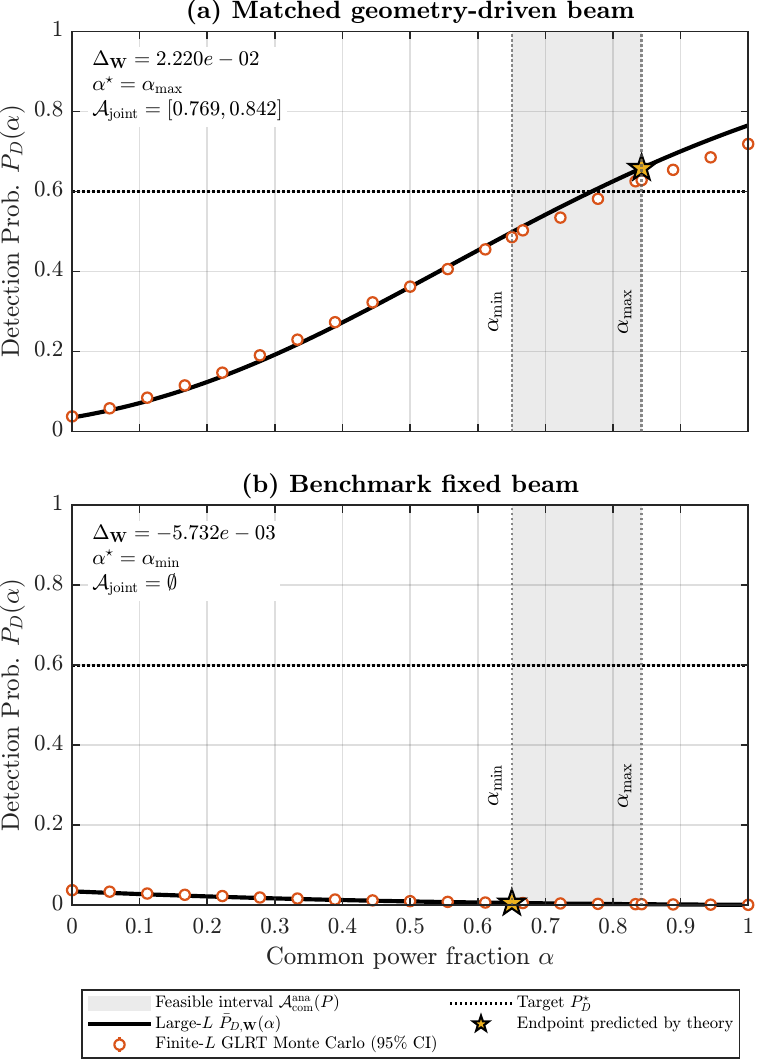}
    \caption{Conditional detection probability over the analytical communication-feasible interval $\mathcal{A}_{\mathrm{comm}}^{\mathrm{ana}}(P)$ at the baseline operating point. Both panels use the same realized private beam set $\mathbf{W}$ and sensing geometry $\mathbf{Q}$. Panel (a) uses the matched beam $\mathbf{w}_c^{\star}=\mathbf{u}_{\max}(\mathbf{Q})$, while panel (b) uses the benchmark fixed beam $\mathbf{w}_c=\mathbf{e}_1$ ($\mathbf{e}_1$ is the first canonical basis vector).}
    \label{fig:detection_audit}
\end{figure}

Fig.~\ref{fig:detection_audit} then tests whether the sensing constraint can be satisfied within the communication-feasible interval. Holding the private beams and sensing geometry fixed; only the common beam is changed. With the matched geometry-driven beam in Fig.~\ref{fig:detection_audit}(a), $\Delta_{\mathbf{W}}>0$, so Corollary~\ref{cor:matched_geometry} selects the upper endpoint $\alpha_{\max}^{\mathrm{ana}}(P)$. That endpoint exceeds the target detection probability, and the finite-$L$ GLRT simulation follows the large-snapshot prediction closely, yielding a nonempty joint-feasible interval near the right edge of $\mathcal{A}_{\mathrm{comm}}^{\mathrm{ana}}(P)$. In contrast, the benchmark fixed beam in Fig.~\ref{fig:detection_audit}(b) gives $\Delta_{\mathbf{W}}<0$, so Theorem~\ref{thm:joint_interval} shifts the best sensing point to the lower endpoint. Even there, the detection probability remains far below the target, and the joint-feasible set is empty. Thus, Fig.~\ref{fig:detection_audit} isolates the main sensing conclusion: for the same communication-feasible interval and private beams, geometry-matched common-beam design makes common-stream-only sensing feasible, whereas an exogenous fixed beam does not.

\section{Conclusion}

This paper addresses a feasibility question in common-stream-only RSMA sensing: when CSIT and SIC are imperfect, can the common stream remain decodable, protect private reliability, and still provide useful sensing energy? The answer shows that the common stream is not automatically a “free” sensing waveform. Our analytical and numerical evaluations reveal a critical structural bottleneck: the combination of imperfect CSIT and SIC enforces a persistent high-SNR "ceiling" on private-stream performance, permanently restricting the maximum common power available for target detection. Furthermore, our derived outage expressions provide a conservative inner approximation of the true communication-feasible interval, enabling robust system screening. The evaluations confirm that utilizing a geometry-matched common beam ensures target detection improves monotonically with the common power fraction. This insight simplifies system design by reducing the complex joint feasibility problem to a single-point certification at the upper communication-feasible boundary. In conclusion, this paper identifies the operational regimes where common-stream-only sensing is viable under non-ideal conditions.

\bibliographystyle{IEEEtran}
\bibliography{references}

\end{document}